\documentclass[]{amsart} 
\usepackage{amsmath}
\usepackage{color}
\usepackage{algorithm}
\usepackage{tikz}
\usepackage{graphicx}
\usepackage{mathtools}
\usepackage{tikz}
\usetikzlibrary{arrows,automata}
\usepackage{amsfonts}
\usepackage{algpseudocode}

\newtheorem{define}{Definition}

\newtheorem{subproblem}{Subproblem}
\newtheorem{problem}{Problem}

\newtheorem{remark}{Remark}
\newtheorem{assumption}{Assumption}
\newtheorem{theorem}{Theorem}
\newtheorem{corollary}{Corollary}
\newtheorem{proposition}{Proposition}

\usepackage{amsfonts}

\begin{document}
\title[Safety Model Predictive Control of Urban Traffic Networks]{A Provably Correct MPC Approach to  Safety Control of Urban Traffic Networks}

%\author{Author One}
%%    Address of record for the research reported here
%\address{Department of Mathematics, Louisiana State University, Baton
%Rouge, Louisiana 70803}
%%    Current address
%\curraddr{Department of Mathematics and Statistics,
%Case Western Reserve University, Cleveland, Ohio 43403}
%\email{xyz@math.university.edu}
%%    \thanks will become a 1st page footnote.
%%    Information for second author
%\author{Author Two}
%\address{Mathematical Research Section, School of Mathematical Sciences,
%Australian National University, Canberra ACT 2601, Australia}
%%    General info
%\subjclass[2000]{Primary 54C40, 14E20; Secondary 46E25, 20C20}
%
%\date{January 1, 2001 and, in revised form, June 22, 2001.}
%\keywords{Differential geometry, algebraic geometry}

\author[Sadra Sadraddini]{Sadra Sadraddini}
\address{Department of Mechanical Engineering, Boston University, Boston, MA 02215}
\email{sadra@bu.edu}
\author[Calin Belta]{Calin Belta}
\address{Department of Mechanical Engineering, Boston University, Boston, MA 02215}
\email{cbelta@bu.edu}
\thanks{This work was partially supported by the NSF under grants CPS-1446151 and CMMI-1400167.}

\begin{abstract}
Model predictive control (MPC) is a popular strategy for urban traffic management that is able to incorporate physical and user defined constraints. However, the current MPC methods rely on finite horizon predictions that are unable to guarantee desirable behaviors over long periods of time. In this paper we design an MPC strategy that is guaranteed to keep the evolution of a network in a desirable yet arbitrary -safe- set, while optimizing a finite horizon cost function. Our approach relies on finding a robust controlled invariant set inside the safe set that provides an appropriate terminal constraint for the MPC optimization problem. An illustrative example is included. 
\end{abstract}

\maketitle

\section{Introduction}
Traffic congestion is a major problem in many cities worldwide. In recent decades, many methods for more efficient usage of existing physical infrastructure have been proposed including new strategies specialized for controlling traffic lights in an urban network \cite{papa_review}. With the advent of new sensing technologies and improvements in online computation capabilities, traffic responsive strategies are gaining more popularity. 

Model predictive control (MPC) is a powerful framework for coordinating urban traffic lights that relies on online optimization while accounting for various constraints in the system \cite{lin2011fast,kamal2014network,mirchand}. However, MPC is known to exhibit ``myopic" behavior that is a result of limited horizon planning. For instance, an MPC traffic light control strategy may lead a network to a state in which undesirable behaviors such as gridlock, spill-back and heavy congestion become inevitable for any future control action. A tempting resolution is elongating the prediction horizon which is often impractical from computational perspective. Some control strategies \cite{mirchand,cesme2013self,christofa2013arterial} have proposed enhancing the control architecture with additional layers that try to detect and avoid undesirable behaviors such as spill-back. However, these methods largely rely on heuristics rather than formal and verifiable measures.

MPC closed-loop strategies that are guaranteed to satisfy a set of constraints are studied extensively in the control theory literature \cite{mayne2000constrained}. Using set invariance theories \cite{blanchini1999survey} and terminal constraints \cite{kerrigan2001}, MPC strategies have been developed that are able to address stability issues while restricting the system trajectory to a convex safe set. However, applying similar methods to urban traffic models is impractical if not impossible due to the complexity of the constraints, controls and uncertainties. Furthermore, a critical bottleneck in the set invariance theories is the inability to deal with non-convex safe sets. Few results exist on computation of non-convex invariant safe sets for linear systems \cite{rakovic2004computation} \cite{perez2011maximal} \cite{rungger2013specification}. Such models are not sufficient for traffic networks, which have to take into consideration finite capacity roads, discrete controls (traffic lights), uncertain exogenous inputs (vehicular arrival rates) and non-convex safe sets that are able to specify desirable behaviors such as avoidance of spill-backs and gridlocks. Furthermore, there exist no results on MPC control of such traffic models from rich, temporal logic \cite{baier2008principles} specifications that are guaranteed to enforce the satisfaction of the specification over long periods of time. 

In this paper, we wish to design an MPC strategy that is guaranteed to confine the evolution of an urban traffic network to a user-defined (non-convex) ``safe" set. Inspired by recent advances in formal methods approaches to control theory, 
we propose a new method to overcome the mentioned issues, which is based on abstracting the network system model to a finite state system. Finite state representations of urban traffic networks have been recently investigated in \cite{coogan2015efficient,coogan2015signal}. In a finite system, we can easily solve a ``safety game" \cite{tabuada2009verification}. We prove that the safe invariant set found in the finite representation corresponds to a robust controlled invariant set in the original system that can be used as a terminal constraint for the MPC optimization problem. We are thus able to guarantee that the evolution of the system remains in the safe set, while planning ahead for optimality. Similar to works in \cite{lin2011fast,kamal2015traffic}, we formulate the MPC optimization problem as a mixed integer linear programming (MILP) problem. We also argue that now being able to use small prediction horizons, we also can rely on full enumeration of possible controls to find the best future control plan instead of solving a possibly large MILP.

 The network model and finite abstraction section of this work is almost identical to the work in  \cite{coogan2015signal}. However, there are clear differences between the control strategies. The authors in \cite{coogan2015signal} find a control strategy that satisfies a linear temporal logic (LTL) specification directly from solving a Rabin game \cite{yordanov2012temporal} on a finite state graph, without any attempt to consider optimality. However, we find an optimization based control strategy in the original continuous system subject to the constraints we find in the abstract finite state system. While LTL control in a traffic network typically requires the user to rigorously specify the desired behavior of each road and intersection, MPC naturally selects a (sub)optimal policy in the absence of user-defined constraints. We also discuss that the methods in this paper can be extended to safety specifications described by bounded time temporal logic formulas such as signal temporal logic (STL) at the expense of higher computational burden.

\section{Traffic Network Model}
\label{sec:model}
%The model used in this paper is adopted from the discrete-time model in \cite{coogan2015signal}.
The urban traffic network model used in this paper is adopted from the discrete-time fluid-like flow model in \cite{coogan2015signal}.
%is based on directed graphs where edges represent finite capacity links (roads) and nodes are signalized intersections. The traffic dynamics is based on the discrete-time fluid-like vehicular flow model from \cite{coogan2015signal}.
The network consists of a set of links denoted by $\mathcal{L}$ and a set of intersections denoted by $\mathcal{V}$. Each $l\in \mathcal{L}$ is a oneway link from \emph{tail} intersection $\tau(l)\in \mathcal{V} \cup \emptyset$ toward \emph{head} intersection $\eta(l)\in \mathcal{V}$. 
The links that flow out of the network are not explicitly modeled.
For each link $l$, the set of \emph{downstream} links is defined as:
\begin{equation}
\mathcal{L}_l^{down}:=\left\{ k\in \mathcal{L} : \tau(k)=\eta(l) \right \}.
\end{equation}
Similarly, the set of \emph{upstream} links is:
\begin{equation}
\mathcal{L}_l^{up}:=\left\{ i\in \mathcal{L} : \eta(i)=\tau(l) \right \},
\end{equation}
and the set of \emph{adjacent} links is:
\begin{equation}
\mathcal{L}_l^{adj}:=\left\{ j\in \mathcal{L} : \tau(j)=\tau(l) \right \}.
\end{equation}
Links in $\mathcal{L}_l^{local}:=\{\mathcal{L}_l^{up} \cup  \mathcal{L}_l^{down} \cup  \mathcal{L}_l^{adj} \}$ are local to $l$.
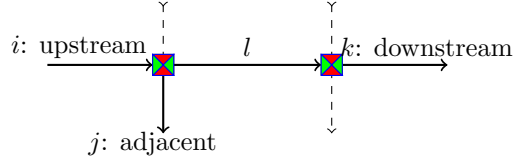
\begin{figure}[t]
\begin{center}
\begin{tikzpicture}[xscale=1.4,yscale=0.7]
\draw [thick,->] (1.2,0.1)--(2.2,0.1);
\draw [thick,->] (2.4,0.1)--(3.8,0.1);
\draw [thick,->] (4,0.1)--(5,0.1);
\draw [thick,->] (2.3,0.0)--(2.3,-1.2);
% The dashed lines
\draw [dashed,-<] (2.3,0.2)--(2.3,1.3);
\draw [dashed,-<] (3.9,0.2)--(3.9,1.3);
\draw [dashed,->] (3.9,0.0)--(3.9,-1.2);
% intersections
\draw [] (2.2,-0.1) rectangle (2.4,0.3);
\draw [] (3.8,-0.1) rectangle (4,0.3);

\draw [blue, fill=red] (2.2,-0.1) -- (2.3,0.1) -- (2.4,-0.1) -- (2.2,-0.1);
\draw [blue, fill=red] (2.2,0.3) -- (2.3,0.1) -- (2.4,0.3) -- (2.2,0.3);
\draw [blue, fill=green] (2.2,-0.1) -- (2.3,0.1) -- (2.2,0.3) -- (2.2,-0.1);
\draw [blue, fill=green] (2.4,-0.1) -- (2.3,0.1) -- (2.4,0.3) -- (2.4,-0.1);

\draw [blue, fill=red] (3.8,-0.1) -- (3.9,0.1) -- (4.0,-0.1) -- (3.8,-0.1);
\draw [blue, fill=red] (3.8,0.3) -- (3.9,0.1) -- (4.0,0.3) -- (3.8,0.3);
\draw [blue, fill=green] (3.8,-0.1) -- (3.9,0.1) -- (3.8,0.3) -- (3.8,-0.1);
\draw [blue, fill=green] (4.0,-0.1) -- (3.9,0.1) -- (4.0,0.3) -- (4.0,-0.1);

\node at (4.8,0.45) {$k$: downstream};
\node at (3.1,0.45) {$l$};
\node at (1.5,0.45) {$i$: upstream};
\node at (2.2,-1.4) {$j$: adjacent};

\end{tikzpicture}

\caption{Links local to link $l$}
\label{fig:local}
\end{center}
\end{figure}
A graphical illustration of the terminologies for a local network is shown in Fig. \ref{fig:local}.
The number of vehicles on  link $l$ at time $t \in \mathbb{Z}_{\geq 0}$ is denoted by $x_l[t] \in [0,x_l^{cap}]$, where $x_l^{cap}$ is the capacity of link $l$. The network state at time $t$ is the vector representation $x[t]=\left \{ x_l \right \}_{l \in \mathcal{L}} \in \mathcal{X} \subset \mathbb{R}^n$ where $n=\left\vert{\mathcal{L}}\right\vert$ is the number of links and 
\begin{equation}
\label{eq:statespace}
\mathcal{X}=\prod_{l\in \mathcal{L}} [0,x_l^{cap}].
\end{equation}
The vehicular flow of link $l$ at time $t$ is controlled by a binary decision denoted by $u_l[t] \in \left\{0,1\right\}$, where values $0$ and $1$ represent the red and green traffic lights, respectively. The control decision combining all traffic lights at time $t$ is $u[t]=\left \{ u_l \right \}_{l \in \mathcal{L}} \in \mathcal{U} \subset \left\{0,1\right\}^n$, where $\mathcal{U}$ is the set of admissible combinations of traffic lights, which is defined with respect to the traffic conventions. For instance, the green/green traffic light combination for two perpendicular links $l_1$ and $l_2$ that have a common head intersection is excluded by adding the constraint $u_{l_1}[t]+u_{l_2}[t] \leq 1$, or $u_{l_1}[t]+u_{l_2}[t] =1$. In the latter case, the red/red combination is also disallowed. Note that the set $\mathcal{U}$ is finite.

%\begin{example}
%TO DO: a simple example of a small network with a graphical illustration.
%\end{example}

%\begin{remark}
%An alternative approach to modeling the relation between the traffic lights is introducing controls as decisions on intersection phases, i.e. the combination of traffic lights. Consecutively, a mapping from intersection phases to decisions on each link can be constructed. This approach is more appropriate for simple intersections, where the number of possible traffic light combinations is small. Furthermore, the approach of modeling constraints as \eqref{eq:agn_links} is more convenient for optimization formulation. 
%\end{remark}

Now we describe the network dynamics. When $u_l[t]=1$, vehicles of link $l$ flow to its downstream links $ \mathcal{L}_l^{down}$. Let $\beta_{lk}$ be the ratio of vehicles turning into link $k \in \mathcal{L}_l^{down}$. The following relation holds for the turning ratios:
\begin{equation}
\sum_{k \in \mathcal{L}_l^{down} } \beta_{lk} \leq 1,
\end{equation}
where the inequality indicates that some vehicles may flow out of the network. The capacity available  at link $l$ to its upstream links at time $t$  is $x_l^{cap}-x_l[t]$. Let $\alpha_{il}^{u[t]}$ be the capacity portion of link $l$ available to link $i \in \mathcal{L}_l^{up}$ when the decision on traffic lights is $u[t]$. The following relation holds for the capacity ratios:
\begin{equation}
\sum_{i \in \mathcal{L}_l^{up} } \alpha_{il}^{u[t]}= 1.
\end{equation}
For simple intersections, it is reasonable to assume $\alpha_{il}^{u[t]}$ is constant if $u_i[t]=1$ and zero otherwise. Therefore, for simplicity, we drop out the ``$u[t]$" superscripts from the capacity ratios in the rest of the paper.
The number of vehicles flowing out of link $l$ at time step $t$ is given by the following equation:
\begin{equation}
\label{eq:flow_out}
f_l[t]=u_l[t] \min \left\{ x_l[t], c_l, \min_{k\in \mathcal{L}_l^{down}} \frac{\alpha_{lk}}{\beta_{lk}} (x_k^{cap}-x_k[t])  \right\},
\end{equation}
where $c_l$ is the maximum number of vehicles that can flow out of link $l$ in one time step. The one-step evolution of link $x_l$ is given by:
\begin{equation}
\label{eq:local}
\begin{split}
& x_l[t+1] =  \mathcal{F}_l (x_l^{local}[t],u[t],d_l[t]) \\ 
&= \min \left\{ x_l[t]-f_l[t]+\underset{i \in \mathcal{L}_l^{up} } \sum \beta_{il} f_i^{}[t] + d_l[t], x_l^{cap} \right\}, \\
\end{split}
\end{equation}
where $x_l^{local}=\left\{x_m\right\}, m=\left\{ l \cup \mathcal{L}_l^{local} \right\}$ and $d_l[t]$ is the number of vehicles arriving in the link $l$  from outside of the network at time $t$.
We also denote $d[t]=\left \{ d_l \right \}_{l \in \mathcal{L}}  \in \mathcal{D} \subset \mathbb{R}^n$, where $\mathcal{D}$ is assumed to be a known set. We observe that $\mathcal{F}_l$ is a piecewise affine function.
The network dynamics is written in the following compact form:
\begin{equation}
\label{eq:dynamics}
x[t+1]=\mathcal{F} \left(x[t],d[t],u[t]\right),
\end{equation} 
where $\mathcal{F} :\mathcal{X} \times \mathcal{D}  \times \mathcal {U} \rightarrow \mathcal{X}$ is a piecewise affine function.

\section{Problem Formulation and Approach}
In this section we formulate the problem and briefly explain the approach. 
First, we define the safety set as a union of hyper-rectangles in the state space $\mathcal{X}$. 
A hyper-rectangle $\mathcal{H} \subset \mathcal{X}$ is defined with respect to a set of \emph{rectangular} inequalities in the form of:
\begin{equation}
\label{eq:hyper_rectangle}
\mathcal{H}:= \left\{ x \in \mathcal{X} \left | \right.
x_{l_i} \le r_{i}
\right \}, i=1,\cdots, p
\end{equation}
where  $r_i \in (0,x_{l_i}^{cap})$ and $p\le n$. A \emph{safe set} $\mathcal{S}$ is defined as a union of hyper-rectangles:
\begin{equation}
\label{eq:safe}
\mathcal{S}:= \bigcup_{s} \mathcal{H}_s , ~s=1,\cdots,n_\mathcal{S},
\end{equation}
where each $ \mathcal{H}_s$ is a hyper-rectangle in the form of \eqref{eq:hyper_rectangle}. Note that the safe set is, in general, non-convex. Considering constraints of the form $x_{l} \le r_l$ stems from the practical purpose that the safe set should always favor fewer number of vehicles on each link. 
For easier readability and expressivity, we can describe $\mathcal{S}$ as the set of values $x \in \mathcal{X}$ that satisfy a boolean expression over a set of rectangular constraints that are connected with boolean operators $\wedge$ (conjunction) and $\vee$ (disjunction). For example, we may express $\mathcal{S}$ as state values that satisfy the boolean expression $\left ( (x_1 \le r_1) \vee (x_2 \le r_2) \right) \wedge (x_3 \le r_3)$. The conversion of a set defined by a boolean expression to the hyper-rectangle form of \eqref{eq:safe} is straightforward and is not discussed in this paper.
\begin{remark}
We can extend the notion of safe set for the instantaneous state into the time evolution of the state. For example, we may desire that if $x_l[t] > r_l$, then $x_l[t+1] \le r_l$. This specification can be captured by the boolean expression $(x_l[t] \le r_l) \vee (x_l[t+1] \le r_l)$, which is over the state space $\mathcal{X} \times \mathcal{X}$. In a more concise way, we are able to construct safe sets in higher dimensions from time bounded temporal logic specifications such as signal temporal logic (STL) and bounded linear temporal logic (BLTL).
\end{remark}

%\begin{example}
%\end{example} 

The network model \eqref{eq:dynamics} is a piecewise affine system controlled by discrete inputs (traffic lights) where exogenous values for $d$ are considered as adversarial inputs. The \emph{safety problem} considered in this paper is finding a control strategy such that for all allowable realizations of adverserial inputs, the evolution of the system remains in the safe set. 
%i.e. $x[t] \in \mathcal{S}, t \in \mathbb{Z}_{\ge 0}$.
 On the other hand, the controls that ensure safety are often not unique. For practical implementation, (sub)optimal selection of the controls subject to an appropriate cost function is also important. Since the system is complex and various constraints are present, we use MPC strategy to the find the optimal control sequence over a finite horizon. Once an optimal control sequence is found, only the current step control is applied to the system and given the new measurements at next time step, a new optimal control sequence is found accordingly. The finite horizon cost criterion considered in this paper is the total time spent (TTS) of the network,  which is used extensively in traffic literature. It can be shown that finite horizon TTS is equivalent to the total number of vehicles \cite{papa_review}:
\begin{equation*}
\label{eq:cost}
J:=\sum \limits_{\tau=t+1}^{t+H} \sum \limits_{l \in \mathcal{L}} x_l[\tau],
\end{equation*}
where $H$ is the length of the prediction horizon. 
The finite horizon control sequence starting at time $t$ is: 
\begin{equation*}
u^H[t]:=\left\{ u[t], u[t+1], \cdots, u[t+H-1]\right\}.
\end{equation*}
Given the finite horizon exogenous input sequence 
 $\left\{ d[t], \cdots, d[t+H-1]\right\}$ and the current system state $x[t]$, one can compute, using \eqref{eq:dynamics}, the finite horizon evolution of the system $\left\{ x[t+1], \cdots, x[t+H]\right\}$. 
Since the values of $d$ are in general unknown, it is impossible to precisely predict the finite horizon cost function. In this paper, we simply replace the values of $d$ with estimated, or nominal, values $d^e[t]$ and predict the cost function with the resulting estimated state values $x^e[t]$. The values of $d^e[t]$ are assumed to be given by some means like online measurements from remote sensing available in a modern urban setting. We also assume that the current state $x[t]$ is known, however, this assumption is not totally restrict as discussed later in the paper. Note that optimizing the estimated, or nominal, cost is often an appropriate approach in robust MPC as other approaches such as optimization of the worst case cost are usually computationally intensive and often result in poor performance \cite{bemporad1999robust}. Notice that the estimated values for adversarial inputs are only used to predict the finite horizon cost whereas the safety property is guaranteed with respect to the all allowable adversarial inputs. 
\vspace{10pt}
\begin{problem}
\label{problem:main}
Given an urban traffic network \eqref{eq:dynamics} and a safety set $\mathcal{S}$ in the form of \eqref{eq:safe}, find a control strategy that for all allowable sequences of adversarial inputs $d[t] \in \mathcal{D}$, the evolution of the system is guaranteed to remain in the safe set:
\begin{equation}
x[t] \in \mathcal{S}, \forall t \in \mathbb{Z}_{\ge 0},
\end{equation}
and, optimize an estimated (nominal) finite horizon cost function:
\begin{equation}
\label{eq:cost}
J^e=\sum \limits_{\tau=t+1}^{t+H} \sum \limits_{l \in \mathcal{L}} x^e_l[\tau],
\end{equation}
where $x^e$ are given by the estimated (nominal) finite horizon evolution of the system.
\end{problem}
\endproof

The MPC optimization problem that includes safety only over the finite horizon  is:
\begin{equation}
\label{eq:optimizeproblem}
\begin{array}{cl}
u^H[t] =& \text{argmin}~ \sum \limits_{\tau=t+1}^{t+H} \sum \limits_{l \in \mathcal{L}} x^e_l[\tau]\\
\text{s.t} & x[\tau] \in  \mathcal{S}, \\
 &  x[\tau+1] = \mathcal{F}(x[\tau],d[\tau],u[\tau]), \\
   & x^e[\tau+1] =\mathcal{F}(x[\tau],d^e[\tau],u[\tau]), \\
 &   \forall d[\tau] \in \mathcal{D}, \tau=t,\dots, t+H-1. \\
\end{array}
\end{equation}

However, finite horizon  safety does not guarantee infinite horizon safety. In other words, it is possible that the MPC optimization problem becomes infeasible at some time. The key contribution of this paper is guaranteeing \emph{recursive feasibility}, which is defined as follows.
\begin{define}
An MPC problem is recursively feasible if the application of control $u[t]$ from the solution of the MPC optimization problem at time $t$ guarantees the feasibility of the MPC optimization problem at time $t+1$. 
\end{define}

A well known method to guarantee recursive feasibility is adding an appropriate terminal constraint \cite{kerrigan2001} to the MPC problem in the form of:
\begin{equation}
\label{eq:terminal}
x[t+H] \in \mathcal{T}, 
\end{equation}
where $\mathcal{T} \subseteq \mathcal{S}$ is the \emph{terminal set}. Our approach to Problem \ref{problem:main} involves solutions to the two following subproblems.

\begin{subproblem}(Terminal Set)
\label{problem:terminal}
Find a terminal set $\mathcal{T}$ such that adding the terminal constraint \eqref{eq:terminal} to the MPC optimization problem \eqref{eq:optimizeproblem} guarantees recursive feasibility. 
\end{subproblem}
\begin{subproblem}(MPC)
\label{problem:mpc}
 Find $u[t]$ by solving the optimization problem \eqref{eq:optimizeproblem} with the addition of the terminal constraint \eqref{eq:terminal}.
\end{subproblem}
Subproblem \ref{problem:terminal} is solved once and in offline fashion, while Subproblem \ref{problem:mpc} is solved online at each time step. It is also reasonable to assume that in the online implementation, more precise knowledge of values of $d$ are available for a finite horizon. However once solving Subproblem \ref{problem:terminal}, the whole set of $\mathcal{D}$ is taken into account.
Our approach to Subproblem \ref{problem:terminal} concerns computing a robust controlled invariant set inside $\mathcal{S}$ that involves abstracting the system into a finite state transition system, which is explained in detail in Section \ref{sec:robust}.
%In Section \ref{sec:control}, we formally state the previously known result \cite{kerrigan2001} that replacing the robust controlled invariant set as a terminal set guarantees recursive feasibility. 
Our solution to Subproblem \ref{problem:mpc} is based on formulating the problem as a MILP problem that is solvable using efficient commercial solvers such as Gurobi and CPLEX. The translation of dynamical and set constraints to MILP is explained in Section \ref{sec:control}.

%Note that, unlike most control architectures, STL control requires keeping the track of the the history of state values up to an extension in the past (which is the formula horizon $h^\varphi$).

\section{Terminal Set and Invariance}
\label{sec:robust}

This section focuses on the solution to Subproblem \ref{problem:terminal}. First, we define the notion of robust controlled invariant set. Next, we provide a summary on how to abstract system \eqref{eq:dynamics} to a finite state transition system. We then state, and prove, that the properties of the abstract system can be used to find a solution to Subproblem \ref{problem:terminal}.

\subsection{Robust Controlled Invariant Set}
\begin{define}
\cite{blanchini1999survey}
Given the discrete time uncertain control system:
\begin{equation}
\label{eq:new_dynamics}
z_{k+1}=f(z_k,v_k,w_k),
\end{equation} 
where $z_k \in \mathcal{Z}$ is the state, $v_k\in \mathcal{V}$ is the control and $w_k \in \mathcal{W}$ is the disturbance or adversarial input,
the set $\mathcal{C} \subseteq \mathcal{Z}$ is \emph{robust controlled invariant} if and only if:
\begin{equation}
\forall z_k \in \mathcal{C}~ \exists v_k \in \mathcal{V} ~s.t.~ f(z_k,v_k,w_k)\in \mathcal{C} ~\forall w_k \in \mathcal{W}.
\end{equation}
\end{define}

It can be shown that the union of robust control invariant sets is itself robust control invariant. It also can be shown that if there exists a robust control invariant set, then there exists a unique \emph{maximal robust controlled invariant set} (MRCIS) which is the union of all possible robust control invariant sets \cite{kerrigan2001}.

\begin{define} \cite{kerrigan2001}
Given a control system \eqref{eq:new_dynamics} and a set $\mathcal{G}$, the robust one step set (robust predecessor) $\Omega(\mathcal{G})$ is defined as:
\begin{equation}
\Omega(\mathcal{G}) = \big \{ z_k \in \mathcal{Z} \left | \right.  \exists v_k \in \mathcal{V} ~s.t.
 ~ f(z_k,v_k,w_k)\in \mathcal{G} ~\forall w_k \in \mathcal{W} \big \}.
\end{equation}
\end{define}
We wish to find the robust controlled invariant set of the traffic network inside the safe set $\mathcal{S}$. The traffic network's MRCIS inside $\mathcal{S}$, which is denoted by $\mathcal{I}$, can be computed using the well known Algorithm \ref{alg:mrcis}\cite{kerrigan2001}.
The main drawback of this algorithm is that finite termination is not guaranteed. Furthermore, early termination results in an over-approximation of MRCIS which is undesirable. Even if MRCIS is determined in finitely many steps, performing the operations in Algorithm \ref{alg:mrcis} for dynamics \eqref{eq:dynamics}, which is piecewise affine with adversarial inputs, and a non-convex safe set $\mathcal{S}$, is computationally intractable due to the severe limitations in performing polyhedral operations like Pontryagin's difference \cite{kerrigan2001} for a non-convex set. Our approach to overcome these issues is to abstract the system as a finite state transition system. In the finite realm, implementation of Algorithm \ref{alg:mrcis} is much easier and finite termination is assured. 

\begin{algorithm}
\caption{Procedure for MRCIS $\mathcal{I}$ inside $\mathcal{S}$}\label{alg:mrcis}
\begin{algorithmic}
   \State $\mathcal{I}_0=\mathcal{S}$
   \While{$\mathcal{I}_{i} \neq \mathcal{I}_{i+1}$ }
      \State $\mathcal{I}_{i+1}=\Omega(\mathcal{I}_i) \cap \mathcal{I}_i$
   \EndWhile
   \State \textbf{return} $\mathcal{I}=\mathcal{I}_{i}$
\end{algorithmic}
\end{algorithm}

\subsection{Finite State Abstraction}
\begin{define}(Finite State Transition System)
A (non-determistic) finite state transition system is defined as the tuple $\mathcal{FTS}=\left( Q,\Sigma, \delta, \Lambda, \lambda \right)$ where $Q$ is a finite set of states, $\Sigma$ is a finite set of symbols (controls), $\delta: Q \times \Sigma \rightarrow 2^ Q$ is the transition function, $\Lambda$ is a finite set of labels and $\lambda: Q \rightarrow \Lambda$ is a labeling function. 
\end{define}
Constructing a finite state transition system $\mathcal{FTS}=\left( Q,\Sigma, \delta, \Lambda, \lambda \right)$ from the original system \eqref{eq:dynamics} is treated in \cite{coogan2015efficient,coogan2015signal} and the details are not presented in this section. Instead, we summarize the main points about the properties of the abstract system and provide the necessary notation for the remaining of the paper. 

Abstraction involves partitioning the state space into a finite set of hyper-rectangles denoted by $Q$. The state space corresponding to each link $l \in \mathcal{L}$, $[0,x_l^{cap}]$, is partitioned into $N_l$ intervals:
\begin{equation}
\label{eq:intervals}
\left \{ [0,x_l^{(1)}],(x_l^{(2)},x_l^{(3)}],\cdots,(x_l^{(N_l-1},x_l^{cap)}] \right \}.
\end{equation}
By performing the cartesian product of the sets from \eqref{eq:intervals} for all different $l \in \mathcal{L}$, $\left | Q\right | = \prod_{l \in \mathcal{L}} N_l$ hyper-rectangles are obtained.
\footnote{We abuse the notation and call these sets hyper-rectangles even though not all of them contain their boundaries (see \eqref{eq:hyper_rectangle}).} 
Each abstract state $q\in Q$ uniquely represents a hyper-rectangle, denoted by $\mathcal{P}(q)$, where the function $\mathcal{P}: Q \rightarrow 2^\mathcal{X}$ is defined to map each abstract state into its hyper-rectangle representation inside $\mathcal{X}$. Note that these sets provide a partition of $\mathcal{X}$:
\begin{equation*}
\bigcup_{q \in Q} \mathcal{P}(q)=\mathcal{X}, ~
 \mathcal{P}(q_a) \cap \mathcal{P}(q_b)  =\emptyset, q_a \neq q_b.
\end{equation*} 
The labeling function $ \Lambda$ is defined with respect to the safe set $\mathcal{S}$. Each $q \in Q$ is labeled as $safe$ if the whole hyper-rectangle is inside $\mathcal{S}$ and $unsafe$ otherwise. Let $Q^\mathcal{S} \subset Q$ represent the set of states that are labeled safe. If the intervals \eqref{eq:intervals} are initialized with respect to $\mathcal{S}$, it can be shown that:
\begin{equation}
\mathcal{S}= \bigcup_{q \in Q^{\mathcal{S}}} \mathcal{P}(q).
\end{equation}
In other words, no $q \in Q$ represents a hyper-rectangle that is only partially in the safe set. 
%We also define the inverse function $\mathcal{R}^{-1}: \mathcal{X} \rightarrow Q$ to find the hyper rectangle $q \in Q$ that $x \in \mathcal{R}(q)$. This function is defined in the  whole but a zero volume subset of state space $\mathcal{X}$ that corresponds to the intersections of the partitions. While computing the set of transitions, as it will be discussed later, it is safe to assume to the inequalities in \eqref{eq:box} as strict inequalities. 

The transitions are determined using dynamics \eqref{eq:dynamics}. Since the set of controls is finite, we let $\Sigma=\mathcal{U}$. For each abstract state $q \in Q$ and control $u \in \mathcal{U}$, the set of one-step reachable abstract states, denoted by $post(q,u)$, is computed by taking into account all allowable adversarial inputs:
\begin{equation}
\begin{array}{ll}
q^\prime \in {post}(q,u)  ~\text{if and only if} 
\\ 
\exists x\in \mathcal{P}(q),  \exists d \in \mathcal{D}   ~ s.t. ~\mathcal{F}(x,u,d) \in \mathcal{P}(q^\prime).
\end{array}
\end{equation}
 The set $post(q,u)$ often includes more than one state which results in non-determinism in the finite state transition system. The computation of $post$ operation for a piece-wise affine system with exogenous inputs requires intensive polyhedral operations. On the other hand, based on the sparsity and component-wise monotonicity properties of the traffic networks, authors in \cite{coogan2015signal,coogan2015efficient} have introduced a computationally efficient method that, under some mild assumptions, slightly over-approximates  $post(q,u)$ by the set $\overline{post}(q,u)$, $post(q,u) \subseteq \overline{post}(q,u)$. Finally, all transition relations $\delta(q,u)=\left\{q^\prime \left | \right. q^\prime \in \overline{post}(q,u) \right \}$ are constructed and the finite state abstraction procedure is completed. 
 \begin{proposition}
\label{prop:sim}
The abstarct finite system \emph{simulates} the original system, i.e. any transition in the original system is captured by at least one transition in the abstract system:
\begin{equation}
\begin{array}{l}
\forall x \in \mathcal{X}~ \forall u\in \mathcal{U}~\forall d \in \mathcal{D} ~s.t.~
 x^\prime=\mathcal{F}(x,u,d),\\
\exists q,q^\prime \in Q , x \in \mathcal{P}(q), x^\prime \in \mathcal{P}(q^\prime) ~s.t. ~ q^\prime \in \delta(q,u).
\end{array}
\end{equation}
\end{proposition}
\begin{proof}
See \cite{coogan2015signal}.
\end{proof}
Note that the simulation property does not state that all trajectories in the abstract system are also present in the original system. In fact, the finite state transition system may include \emph{spurious} trajectories that are not present in the original system. Although the presence of these transitions do not affect safety control, they introduce conservatism in optimal planning.

\subsection{MRCIS for the Abstract System}
In this section, we find the abstract system's MRCIS, denoted by $Q^{\mathcal{I}}$, in the safe set $Q^{\mathcal{S}}$. The analogous of Algorithm \ref{alg:mrcis} for the MRCIS of a finite system  is also known as the ``safety game" \cite{tabuada2009verification}, which is based on iteratively removing the states that are absorbed into the unsafe set for all controls. The procedure is outlined in Algorithm \ref{alg:fmrcis}.
\begin{algorithm}
\caption{Procedure for MRCIS $ Q^{\mathcal{I}}$ inside $ Q^\mathcal{S}$}\label{alg:fmrcis}
\begin{algorithmic}[1]
   \State $Q^{\mathcal{I}}= Q^\mathcal{S}$
   \While{$ Q^{a} \neq \emptyset$ }
      \State $ Q^{a}=\emptyset$
      \For {$ q \in  Q^{\mathcal{I}}$}
     		\If  { $\forall u\in \mathcal{U} ~\delta( q, u) \not \subseteq  Q^{\mathcal{I}}$  }
			\State {$ Q^a \gets  Q^a \cup  q$}
      \EndIf
      \EndFor
      \State $ Q^{\mathcal{I}} \gets  Q^{\mathcal{I}} \setminus  Q^{a}$
   \EndWhile
   \State \textbf{return} $ Q^\mathcal{I}$
\end{algorithmic}
\end{algorithm}
We define the hyper-rectangle representation of $Q^\mathcal{I}$ as
\begin{equation}
\mathcal{\tilde I}:= \bigcup_{q \in \mathcal{Q}^\mathcal{I}} \mathcal{P}(q).
\end{equation}
Since $Q^\mathcal{I} \subseteq Q^\mathcal{S}$, it is clear that $\mathcal{\tilde I} \subseteq \mathcal{S}$.
\begin{theorem}
$\mathcal{\tilde I}$ is a robust controlled invariant set of the original system \eqref{eq:dynamics}.
\end{theorem}
\begin{proof}(proof by contradiction)
Suppose that $\mathcal{\tilde I}$ is not a robust controlled invariant set. Then 
$$\exists x \in \mathcal{\tilde I}~s.t.~ \forall u \in \mathcal{U}~\exists d \in \mathcal{D}~s.t.~ x^\prime=\mathcal{F}(x,u,d) \not\in \mathcal{I}.$$
In words, there exists a state $x$ such that for all available controls, there exist an exogenous input that the next step state value $x^\prime$ is outside of the set $\mathcal{\tilde I}$.
Let $q$ and $q^\prime$ be the abstract states that $x \in \mathcal{P}(q), x^\prime \in \mathcal{P}(q^\prime)$. 
The above expression is equivalent to: 
\begin{equation*}
\begin{array}{l}
\exists q \in Q^\mathcal{I} ~s.t.~\exists x \in \mathcal{P}(q) ~s.t.~ \forall u \in \mathcal{U}~\exists d \in \mathcal{D}~ \\
 ~ s.t. ~ \exists q^\prime \not \in Q^\mathcal{I}  ~ s.t. ~ \exists x^\prime \in \mathcal{P}(q^\prime), x^\prime=\mathcal{F}(x,u,d).
\end{array}
\end{equation*}
We claim that $q^\prime \in \delta(q,u)$. If $q^\prime \not \in \delta(q,u)$, then there exists a state $x \in \mathcal{P}(q)$ such that $\exists u \in \mathcal{U}, \exists d \in \mathcal{D}$ that $x^\prime = \mathcal{F}(x,u,d)$ and $x^\prime \in \mathcal{P}(q^\prime)$ where $q^\prime \not \in \delta(q,u)$, which violates the simulation property (Proposition \ref{prop:sim}). Therefore, $q^\prime \in \delta(q,u)$ hence
\begin{equation*}
\exists q \in Q^\mathcal{I} ~s.t. ~\forall u\in \mathcal{U}~ \exists q^\prime \not \in Q^\mathcal{I} ~s.t.~ q^\prime \in \delta(q,u),
\end{equation*} 
which is equivalent to
\begin{equation*}
\exists q \in Q^\mathcal{I}~ s.t. ~ \forall u\in \mathcal{U} ~\delta( q, u) \not \subseteq  Q^{\mathcal{I}},
\end{equation*}
that is in contradiction with lines 5 to 9 of Algorithm \ref{alg:fmrcis} that states such an abstract state is removed from $ Q^\mathcal{I}$. Therefore, $\mathcal{\tilde I}$ is a robust controlled invariant set. 
\end{proof}
\begin{corollary}
The abstract system's MRCIS is a subset of the original system's MRCIS. i.e.
$\mathcal{\tilde I} \subseteq \mathcal{I}.
$
\end{corollary}
\begin{proof}
Since $\mathcal{I}$ is a the maximum robust controlled invariant set and unique, $\mathcal{\tilde I} \subseteq \mathcal{I}
$. 
\end{proof}

\subsection{Recursive Feasibility}
Now we prove that replacing the terminal set with a robust controlled invariant set guarantees recursive feasibility.  
\begin{theorem}
The following MPC optimization problem:
\begin{equation}
\label{eq:optimterminal}
\begin{array}{cl}
u^{H,opt}[t] =& \text{argmin}~  \sum \limits_{\tau=t+1}^{t+H} \sum \limits_{l \in \mathcal{L}} x^e_l[\tau]
\\
\text{s.t} & x[\tau] \in  \mathcal{S}, \\
& x[t+H] \in \mathcal{\tilde I}, \\
 &  x[\tau+1] = \mathcal{F}(x[\tau],d[\tau],u[\tau]), \\
   & x^e[\tau+1] =\mathcal{F}(x[\tau],d^e[\tau],u[\tau]), \\
 &   \forall d[\tau] \in \mathcal{D}, \tau=t,\dots, t+H-1. \\
\end{array}
\end{equation}
is recursively feasible.
\end{theorem}
\begin{proof}
Suppose that the solution to the optimization problem above is found as $u^{H,opt}[t]=\left \{ u^{t}[t], u^{t}[t+1], \cdots, u^{t}[t+H-1] \right \}$, where $u^{t}[t]$ is applied to the system.  
% and the optimal state sequence is $\left \{ x^{opt}[t+1], x^{opt}[t+2], \cdots, x^{opt}[t+H] \right \}$
We need to prove that the feasible set of the optimization problem of $u^H[t+1]$ is nonempty. It is known from the optimization problem of $u^H[t]$ that applying the rest of the open-loop sequence $\left \{ u^{t}[t+1], u^{t}[t+2]\cdots, u^{t}[t+H-1] \right \}$ will result $x[t+H]$ to be inside $\mathcal{\tilde I}$. Since $\mathcal{\tilde I}$ is a robust controlled invariant set, for any $x[t+H] \in \mathcal{\tilde I}$ there exists at least at least one control $\tilde u$ that guarantees $\mathcal{F}(x[t+H],\tilde u,d) \in \mathcal{\tilde I}~ \forall d \in \mathcal{D}$. Therefore, $u^H[t+1]=\left \{ u^{t}[t+1], u^{t}[t+2]\cdots, u^{t}[t+H-1], \tilde u \right \}$ satisfies the constraints of optimization problem of $u^H[t+1]$ which hence guarantees recursive feasibility. 
\end{proof}
Finally, we have found a solution to Problem \ref{problem:terminal} and we let the terminal constraint to be:
\begin{equation}
\mathcal{T}=\mathcal{\tilde I}.
\end{equation}

\subsection{Discussions}
\label{sec:discussion1}
Algorithm \ref{alg:fmrcis} may end with an empty set which, unfortunately, does not conclude that $\mathcal{I}$ is also empty. In this case, finer partitions are required to find a nonempty invariant set. Furthermore, refinement may also enlarge $\mathcal{\tilde I}$ if it is already nonempty, which leads to less conservative hence more control options. Note that the refinement of partitions outside $Q^\mathcal{S}$ is not necessary. However, the number of partitions and transitions in the abstract system grows exponentially with respect to the network size. Therefore, our approach to safety control is restricted to small networks.

We need to also address an important issue regarding initial feasibility of the MPC optimization problem. Since the controls are open-loop, no initially feasible problem for large horizons may exist. In the most extreme case, an initially feasible optimization problem may be only available for $H=1$ (where the trajectories should always remain in the invariant set $\mathcal{\tilde I}$). Otherwise, if the initial MPC optimization problem is infeasible, for instance $x$ is initially outside $\mathcal{S}$, one can solve the ``reachability" game  \cite{tabuada2009verification} to guide the state of the system into $\mathcal{\tilde I}$ and later start implementing the MPC. It is worth to note that the reachability game may also be infeasible from some initial conditions.

\section{Model Predictive Control}
\label{sec:control}
In this section, we provide the solution to Subproblem \ref{problem:mpc} based on MILP formulation of the optimization problem \eqref{eq:optimterminal}. We also discuss the limitations of the MILP approach.
\subsection{Mixed Logical Representations of Traffic Networks}
The traffic network dynamics \eqref{eq:dynamics} is a hybrid system that falls into the class of mixed logical dynamical (MLD) systems which can be encoded using mixed integer constraints \cite{bemporad1999control}. Formally, we have:
\begin{proposition}
The traffic dynamics \eqref{eq:dynamics} can be formulated as a finite set of mixed integer constraints:
\begin{equation}
\label{eq:mixed}
\begin{array}{ll}
x[t+1]=  & \mathcal{F} \left(x[t],d[t],u[t]\right) \\
\Leftrightarrow  & x[t+1]+ E_x x[t] + E_u u[t]  + E_d d[t] + E_z z[t] + E_\delta \delta[t] \le e,
\end{array}
\end{equation}
where $z[t]$ is the vector of auxiliary continuous variables, $\delta[t]$ is the vector of auxiliary binary variables, $E_x$, $E_u$, $E_d$, $E_z$, $E_\delta$ are appropriately defined constant matrices and $e$ is a vector that is defined such that  the set of mixed-integer constraints is \emph{well posed}, i.e. for given values of $x[t],u[t]$ and $d[t]$, the feasible set of $x[t+1]$ is a single point. 
\end{proposition}
\begin{proof}
The proof is constructive. 
We define the auxiliary continuous variables $z_l[t], l \in \mathcal{L}$: 
\begin{equation}
z_l[t]:= \min \left\{ x_l[t], c_l, \min_{k\in \mathcal{L}_l^{down}} \frac{\alpha_{lk}}{\beta_{lk}} (x_k^{cap}-x_k[t])  \right\}.\end{equation}
Eqn. \eqref{eq:flow_out} is written as:
\begin{equation}
\label{eq:y}
f_l[t]=u_l[t] z_l[t],
\end{equation}
where the integer constraints encoding $z_l[t]$ are:
\begin{equation}
\label{eq:outflow}
\left\{
\begin{array}{l}
z_l[t] \le x_l[t], z_l[t] \ge x_l[t] - M \delta_{l,1},
\\
z_l[t] \ge c_k - M \delta_{l,2}, z_l[t] \le c_l,
\\
z_l[t] \le \frac{\alpha_{lk}}{\beta_{lk}} (x_l^{cap}-x_l[t]), \forall k \in \mathcal{L}_l^{down} ,
\\
z_l[t] \ge \frac{\alpha_{lk}}{\beta_{lk}} (x_l^{cap}-x_l[t]) -M \delta_{l,k},\forall k \in \mathcal{L}_l^{down},\\
\delta_{l,i} \in \{0,1\},
\\
\delta_{l,1} + \delta_{l,2} + \displaystyle \sum \limits_{k \in \mathcal{L}_l^{down}} \delta_{l,k} = 2+\left | \mathcal{L}_k^{down} \right |,
\\
\end{array}
\right.
\end{equation}
where $M$ is a sufficiently large constant such that
 $$
M \geq \frac{\alpha_{kl}}{\beta_{kl}} x_l^{cap}~~ \forall l \in \mathcal{L}, k \in  \mathcal{L}_l^{down},
$$
and $\delta_{l,i}$'s are auxiliary binary variables. 
The relations above ensure that the flow satisfies the $\min$ argument in \eqref{eq:flow_out}. Next, relation \eqref{eq:y} is encoded as:
\begin{equation}
\left\{
\begin{array}{l}
f_l[t] \le c_l u_l[t],
 f_l[t] \ge 0,
\\
f_l[t] \le z_l[t],
f_l[t] \ge z_l[t] - c_l u_l[t],
\\
u_l[t] \in \left \{0,1 \right\},
\end{array}
\right.
\end{equation}
where $c_l \ge z_l[t]$ is enforced in \eqref{eq:outflow}. We also define
\begin{equation}
z^\prime_l[t]:= x_l[t]-f_l[t]+\underset{i \in \mathcal{L}_l^{up} } \sum \beta_{il} f_i^{}[t] + d_l[t],
\end{equation}
which is linear in terms introduced before. Finally
\begin{equation}
\left\{
\begin{array}{l}
\mathcal{F}_l[t] \le z_l^\prime[t],
 \mathcal{F}_l[t] \le x_l^{cap},
\\
\mathcal{F}_l[t] \ge z_l^\prime[t] - M^\prime \delta_{z_l^\prime},
\\
\mathcal{F}_l[t] \ge x_l^{cap}-M^\prime (1-\delta_{z_l^\prime}),
\\
\delta_{z^\prime} \in \left \{0,1\right \},
\end{array}
\right.
\end{equation}
where $M^\prime \ge \max \limits_{l\in \mathcal{L}} x_l^{cap} $ is a sufficiently large constant. Since $\mathcal{U}$ is also encoded by integer constraints (as explained in Sec. \ref{sec:model}), all the dynamical constraints can be described as mixed integer constraints that in a compact form are written as \eqref{eq:mixed}.
\end{proof}

\subsection{Robust MPC}
A solution to the MPC optimization problem requires that the safety and terminal constraints are satisfied for all allowable adversarial inputs $d \in \mathcal{D}$. In this section, we briefly explain how to characterize the $H$-step reachable sets.  
\subsubsection{One-step reachable set}
\label{sec:one}
We assume the set $\mathcal{D}$ is given as a union of hyper-rectangles:
\begin{equation}
\mathcal{D} = \bigcup_{i_d=1} D_{i_d}, i_d=1,\cdots, n_\mathcal{D},
\end{equation}
such that each $D_{i_d}$ is a hyper-rectangle in the form of  
$\left \{ d \left | \right. \underline{d}^{i_d} \le d \le \overline{d}^{i_d} \right \},$
 where the inequalities are interpreted element-wise. Note that this assumption is not restrictive as one may over-approximate any bounded set by hyper-rectangles. Given control $u[\tau]$, we wish to compute the one-step reachable set of hyper-rectangle $\prod_{l \in \mathcal{L}} \left[\underline{x}_l[\tau], \overline{x}_l[\tau]\right]$.
 \begin{assumption}
\label{assume:c}
For all $ l \in \mathcal{L}, i \in \mathcal{L}_l^{up}$, the following inequality holds:
\begin{equation}
x_l^{cap} \geq c_l + c_i \frac{\beta_{il}}{\alpha_{il}}.
\end{equation}
\end{assumption}

The assumption above is satisfied if the maximum flow rates are small enough, which is physically accomplished when the length of the discrete time steps are small. It is shown in \cite{coogan2015signal} that, under Assumption \ref{assume:c}, an increase in $x_l[t]$ leads to an increase in $x_l[t+1]$. Furthermore, $x_l[t+1]$ is monotonically increasing (decreasing) with respect to the number of vehicles on its downstream and upstream (adjacent) links. Using this \emph{component-wise monotonicity} property, we have:
\begin{equation}
\label{eq:over}
\begin{array}{ll}
\overline{x}^{i_d}_l[\tau+1]= & \max \Big\{ \overline{x}_l^{i_d}[\tau]- 
 u_l[\tau] \min \big \{ \overline{x}_l^{i_d}[\tau], {c}_l, \underset{l\in \mathcal{L}_k^{down}} \min  \frac{{\alpha}_{lk}}{{\beta}_{lk}} ({x}_k^{cap}-\overline{x}^{i_d}_k[\tau])  \big \} \\ &
 + \underset{ i \in \mathcal{L}_k^{up} } \sum {\beta}_{il} u_i[\tau] \min \Big \{  \overline{x}_i^{i_d}[\tau],{c}_i, \frac{{\alpha}_{il}}{{\beta}_{il}} ({x}_l^{cap} 
-\overline{x}^{i_d}_l[\tau]), \\ &
  \underset{j\in \mathcal{L}_l^{adj}} \min \frac{{\alpha}_{ij}}{{\beta}_{ij}} ({x}_j^{cap}-\underline{x}_j^{i_d}[\tau]) \big \}  + \overline{d}_l^{i_d}[\tau], x_l^{cap} \Big \},
\end{array}
\end{equation}
%\begin{equation}
%\label{eq:over}
%\begin{split}
%\overline{x}^{s}_l[\tau+1]= & \max \Bigg\{ x_l^{cap}, \overline{x}^{s}_l[\tau]- 
%\\
% & u_l[\tau] \min \left\{ \overline{x}_l[\tau], {c}_l, \min_{k\in \mathcal{L}_l^{down}} \frac{{\alpha}_{lk}}{{\beta}_{lk}} ({x}_k^{cap}-\overline{x}_k^{i_d}[t])  \right\}
%\\ 
%+ & \underset{ i \in \mathcal{L}_k^{up} } \sum {\beta}_{il} u_i[t] \min \left \{ \overline{x}^{i_d}_i[\tau],{c}_i, \frac{{\alpha}_{il}}{{\beta}_{il}} ({x}_l^{cap}-\overline{x}^{i_d}_l[\tau]) , \right. 
%\\
%&\left. \left. \underset{j\in \mathcal{L}_l^{adj}} \min \frac{{\alpha}_{ij}}{{\beta}_{ij}} ({x}_j^{cap}-\underline{x}^{i_d}_j[\tau]) \right\} + \overline{d}_l^{i_d}[\tau] \right\},
%\end{split}
%\end{equation}
\begin{equation}
\label{eq:under}
\begin{array}{ll}
\underline{x}^{i_d}_l[\tau+1]=  & \max \Big\{ \underline{x}_l^{i_d}[\tau]- 
 u_l[\tau] \min \big \{ \underline{x}_l^{i_d}[\tau], {c}_l, \underset{l\in \mathcal{L}_k^{down}} \min  \frac{{\alpha}_{lk}}{{\beta}_{lk}}  ({x}_k^{cap}-\underline{x}^{i_d}_k[\tau])  \big \}  + \\
&   \underset{ i \in \mathcal{L}_k^{up} } \sum {\beta}_{il} u_i[\tau] \min \Big \{ \underline{x}_i^{i_d}[\tau],{c}_i, \frac{{\alpha}_{il}}{{\beta}_{il}} ({x}_l^{cap} 
-\underline{x}^{i_d}_l[\tau]), \\ &
  \underset{j\in \mathcal{L}_l^{adj}} \min \frac{{\alpha}_{ij}}{{\beta}_{ij}} ({x}_j^{cap}-\overline{x}_j^{i_d}[\tau]) \big \}  + \underline{d}_l^{i_d}[\tau], x_l^{cap} \Big \},
\end{array}
\end{equation}
%
%\begin{tiny}
%\begin{equation}
%\label{eq:under}
%\begin{array}{l}
%\underline{x}^{i_d}_l[\tau+1]=  \max \Big\{ \underline{x}_l^{i_d}[\tau]- 
% u_l[\tau] \min \left\{ \underline{x}_l^{i_d}[\tau], {c}_l, \underset{l\in \mathcal{L}_k^{down}} \min  \frac{{\alpha}_{lk}}{{\beta}_{lk}} ({x}_k^{cap}-\underline{x}^{i_d}_k[\tau])  \right\}
%\\ 
%+  \underset{ i \in \mathcal{L}_k^{up} } \sum {\beta}_{il} u_i[\tau] \min \left \{ \underline{x}_i^{i_d}[t],{c}_i, \frac{{\alpha}_{il}}{{\beta}_{il}} ({x}_l^{cap}-\underline{x}^{i_d}_l[\tau]),
%  \underset{j\in \mathcal{L}_l^{adj}} \min \frac{{\alpha}_{ij}}{{\beta}_{ij}} ({x}_j^{cap}-\overline{x}_j[\tau]) \right\} \\ + \left.\underline{d}_l^{i_d}[\tau], x_l^{cap} \right\},
%\end{array}
%\end{equation}
%\end{tiny}
%
%\begin{equation}
%\label{eq:under}
%\begin{split}
%\underline{x}^{i_d}_l[\tau+1]= & \max \Bigg\{ x_l^{cap}, \underline{x}_l^{i_d}[\tau]- 
%\\
% & u_l[\tau] \min \left\{ \underline{x}_l^{i_d}[\tau], {c}_l, \min_{l\in \mathcal{L}_k^{down}} \frac{{\alpha}_{lk}}{{\beta}_{lk}} ({x}_k^{cap}-\underline{x}^{i_d}_k[\tau])  \right\}
%\\ 
%+ & \underset{ i \in \mathcal{L}_k^{up} } \sum {\beta}_{il} u_i[t] \min \left \{ \underline{x}_i^{i_d}[t],{c}_i, \frac{{\alpha}_{il}}{{\beta}_{il}} ({x}_l^{cap}-\underline{x}^{i_d}_l[\tau]) , \right. 
%\\
%&\left. \left. \underset{j\in \mathcal{L}_l^{adj}} \min \frac{{\alpha}_{ij}}{{\beta}_{ij}} ({x}_j^{cap}-\overline{x}_j[\tau]) \right\} + \underline{d}_l^{i_d}[\tau] \right\},
%\end{split}
%\end{equation}
where lower (over) bars stand for lower (upper) bounds of the values $x$ and $d$, and superscript $i_d$ stands for the values obtained from exogenous inputs in hyper-rectangle $\mathcal{D}_{i_d}$. We denote:
\begin{equation}
\mathcal{R}_{i_d}(\left[\underline x[\tau],\overline x[\tau]\right],u[\tau])=\prod_{l \in \mathcal{L}} \left[\underline{x}_l^{i_d}[\tau+1],\overline{x}_l^{i_d}[\tau+1]\right].
\end{equation}
The one step reachable set is thus given by an union of hyper-rectangles:
\begin{equation}
\label{eq:one}
\mathcal{R}^{(1)}(\underline x[\tau],\overline x[\tau],u[\tau])=\bigcup_{i_d=1}^{n_\mathcal{D}} \mathcal{R}_{i_d}(\left[\underline x[\tau],\overline x[\tau]\right],u[\tau]).
\end{equation}

\subsubsection{$H$-step reachable set}
Using \eqref{eq:one}, the $H$-step reachable is given by:
\begin{small}
\begin{equation}
\label{eq:hstep}
\mathcal{R}^{(H)}([\underline x[t],\overline x[t]],u^H[t])=
\bigcup_{{i^H_d}=1}^{n_\mathcal{D}} {\tiny{\ldots}} \bigcup_{{i^1_d}=1}^{n_\mathcal{D}}\mathcal{R}_{{i^H_d}} \left (\left [\ldots \left [  \mathcal{R}_{{i^1_d}} (\left[\underline x[t],\overline x[t]\right],u[t])  \right] \ldots \right ] , u[t+H-1]\right),
\end{equation}
\end{small}
where $\underline x[t]=\overline x[t]$. Note that we may assume bounded uncertainties in online measurements of $x[t]$ but this may cause issues with recursive feasibility guarantee for the MPC problem
 \footnote{In short, uncertainty in $x[t]$ can be considered up to a level that the current abstract state $q$, which $x[t] \in \mathcal{P}(q)$, can be uniquely determined.}.
The number of mixed logical equations in the form of \eqref{eq:mixed} required to encode the $H$-step reachable set is $2 n_\mathcal{D}^H$ (one for each set of \eqref{eq:over}, \eqref{eq:under}). The MILP problem becomes quickly intractable if $n_\mathcal{D}>1$. Therefore, MILP implementation requires $\mathcal{D}$ to be given by a single hyper-rectangle.  
\begin{remark}
The network dynamics is \emph{monotonic} if no link in the model has an adjacent link. Many models for single arterials with side streets fall into this category. In this case, only the upper bounds $\overline{x}_l$ and $\overline{d}$ are required to compute the $H$-step evolution, which significantly reduces the MILP size.  
\end{remark}

\subsection{Non-Convex Sets and Additional Integer Constraints}
The constraints corresponding to the safe set $\mathcal{S}$ and terminal set $\mathcal{T}$ also can be encoded as mixed integer constraints. We do not explain this method as encoding non-convex sets using integer constraints is a well known straightforward procedure \cite{bertsimas1997introduction}. The main issue  is that the number of integer constraints describing the sets can be very large, and hence the MILP formulation becomes impractical. One approach to overcome this issue is considering the set constraints as \emph{lazy constraints} \cite{optimization2014inc}, i.e. adding them to the MILP problem if they are violated by the relaxed solution. However, this approach may still require incorporating all the constraints.

\subsection{Discussions}
The computational complexity of MILPs grow exponentially with respect to the number of integer constraints. Therefore, practical implementation of MILP in our framework is restricted to simple problems. 

Since the set of controls is finite, a much simpler approach is full enumeration of the all controls over $H$-step horizon. This approach is computationally tractable if the size of $\mathcal{U}$ and the horizon $H$ are small, which is the case in small networks. Using full enumeration, we may consider arbitrarily complex safe sets. 

\section{Example}
We implemented our methods on a simple urban traffic network. The results are presented in this section. Consider a network with $9$ links and $3$ intersections, as illustrated in Fig. \ref{fig:network}. The parameters of the model are given as
\begin{equation*}
\begin{array}{c}
x_1^{cap}=x_2^{cap}=x_3^{cap}=x_4^{cap}=x_5^{cap}=x_6^{cap}=55, \\
x_8^{cap}=x_9^{cap}=40, c_7=c_7=c_9=15\\
c_1=c_2=c_2=c_4=c_5=c_6=20, \\ 
\beta_{12}=\beta_{23}=\beta_{45}=\beta_{56}=0.7,\beta_{72}=0.5,\\
\beta_{86}=\beta_{83}=0.4,\beta_{95}=0.3, 
\end{array}
\end{equation*}
and all capacity ratios are one. 
The safety set is given by the boolean expression 
$\phi_1 \wedge \phi_2 \wedge \phi_3 \wedge \phi_4,
$ where
\begin{equation*}
\begin{array}{c}
\phi_1=(x_1 \le 36) \wedge (x_4 \le 36),
\\
\phi_2=(x_2 \le 44) \vee (x_3 \le 44),
\\
\phi_3=(x_5 \le 44) \vee (x_6 \le 44),
\\
\phi_4=(x_7 \le 32) \vee (x_8 \le 32) \vee (x_9 \le 32).
\end{array}
\end{equation*}
In plain English, $\phi_1$ requires that the entry arterials $1,4$ are never congested, $\phi_2$ and $\phi_3$  state that if the  traffic on a link in the mid-corridor is heavy, the other is light.  $\phi_4$ prevents the entry side links $7,8,9$ from being congested simultaneously. We assume that each intersection has two modes of controls corresponding to the horizontal and vertical flows, $\left |\mathcal{U} \right|=2^3=8$. The set characterizing exogenous inputs (arrival rates) is $\mathcal{D}=\left\{d \left | \right.  {\bf{0}} \le d \le (15,0,0,15,0,0,10,10,10)^T\right\}$. 

We abstract the states corresponding to the  links $1,4,7,8,9$ to $3$ intervals and the states of the remaining links into $2$, generating $3888$ abstract states, where $1664$ of them represent the safe hyper-rectangles. Computing the finite state transition system took $434$ seconds on a 3 GHz dual core Macbook Pro. Solving the safety game took $31$ seconds and the finite state transition system's MRCIS, $Q^\mathcal{I}$, consists of $1176$ abstract states. 

We solved the robust MPC with horizon $H=3$ relying on full enumeration of the controls. We used \eqref{eq:hstep} to find $3$-step reachable set and checked the safety and the terminal constraint for all the reachable set. We didn't use MILP as the number of constraints required to encode the safe and terminal sets are very large. However, we plan to investigate the problem of (approximately) encoding the set constraints using minimal number of constraints in future. Fig. \ref{fig:example} b) shows the \emph{robustness} of the trajectory obtained from the MPC solution simulated for 20 time steps. The robustness is computed by measuring the minimum Euclidian distance of the system's state to the safety set's boundaries \footnote{The robustness defined in this section is inspired by the definition of STL robustness (see \cite{maler_stl}).}. The exogenous inputs were randomly chosen from $\mathcal{D}$ with uniform distribution. The values of estimated exogenous inputs for optimal planning were also randomly chosen.

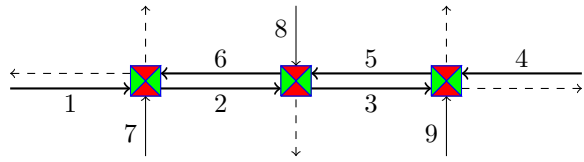
\begin{figure}[t]
\begin{center}
\begin{tikzpicture}[xscale=2,yscale=1]
\draw [thick,->] (0.1,-0.1)--(0.9,-0.1);
\draw [thick,->] (1.1,-0.1)--(1.9,-0.1);
\draw [thick,->] (2.1,-0.1)--(2.9,-0.1);
\draw [dashed,->] (3.1,-0.1)--(3.9,-0.1);
% The other side
\draw [dashed,<-] (0.1,0.1)--(0.9,0.1);
\draw [thick,<-] (1.1,0.1)--(1.9,0.1);
\draw [thick,<-] (2.1,0.1)--(2.9,0.1);
\draw [thick,<-] (3.1,0.1)--(3.9,0.1);
% The junctions
\draw [] (0.9,-0.2) rectangle (1.1,0.2);
\draw [] (1.9,-0.2) rectangle (2.1,0.2);
\draw [] (2.9,-0.2) rectangle (3.1,0.2);
% The side streets
\draw [->] (1,-1)--(1,-0.2);
\draw [dashed,<-] (2,-1)--(2,-0.2);
\draw [->] (3,-1)--(3,-0.2);
\draw [dashed,<-] (1,1)--(1,0.2);
\draw [->] (2,1)--(2,0.2);
\draw [dashed,<-] (3,1)--(3,0.2);

\draw [blue, fill=red] (0.9,-0.2) -- (1.1,-0.2) -- (1,0) -- (0.9,-0.2);
\draw [blue, fill=red] (0.9,0.2) -- (1.1,0.2) -- (1,0) -- (0.9,0.2);
\draw [blue, fill=green] (0.9,-0.2) -- (1,0) -- (0.9,0.2) -- (0.9,-0.2);
\draw [blue, fill=green] (1.1,-0.2) -- (1,0) -- (1.1,0.2) -- (1.1,-0.2);

\draw [blue, fill=red] (1.9,-0.2) -- (2.1,-0.2) -- (2,0) -- (1.9,-0.2);
\draw [blue, fill=red] (1.9,0.2) -- (2.1,0.2) -- (2,0) -- (1.9,0.2);
\draw [blue, fill=green] (1.9,-0.2) -- (2,0) -- (1.9,0.2) -- (1.9,-0.2);
\draw [blue, fill=green] (2.1,-0.2) -- (2,0) -- (2.1,0.2) -- (2.1,-0.2);

\draw [blue, fill=red] (2.9,-0.2) -- (3.1,-0.2) -- (3,0) -- (2.9,-0.2);
\draw [blue, fill=red] (2.9,0.2) -- (3.1,0.2) -- (3,0) -- (2.9,0.2);
\draw [blue, fill=green] (2.9,-0.2) -- (3,0) -- (2.9,0.2) -- (2.9,-0.2);
\draw [blue, fill=green] (3.1,-0.2) -- (3,0) -- (3.1,0.2) -- (3.1,-0.2);

\node at (0.5,-0.3) {1};
\node at (1.5,-0.3) {2};
\node at (2.5,-0.3) {3};
\node at (3.5,0.3) {4};
\node at (1.5,0.3) {6};
\node at (2.5,0.3) {5};

\node at (0.9,-0.7) {7};
\node at (1.9,0.7) {8};
\node at (2.9,-0.7) {9};

%\node at (3.1,-0.25) {5};
%\node at (4.5,-0.25) {7};
%\node at (5.5,0.45) {2};
%\node at (4.5,0.45) {4};
%\node at (3.1,0.45) {6};
%\node at (1.7,0.45) {8};
%
%\node at (3.8,0.8) {c};
%
%\node at (2.2,-0.6) {b};
%\node at (5,-0.6) {d};

\end{tikzpicture}

\caption{The urban traffic network studied in this example} 
\label{fig:network}
\end{center}
\end{figure}

\begin{figure}[t]
\begin{center}
\begin{tabular}{@{}c@{}@{}c@{}}

            \includegraphics[width=0.49\textwidth]{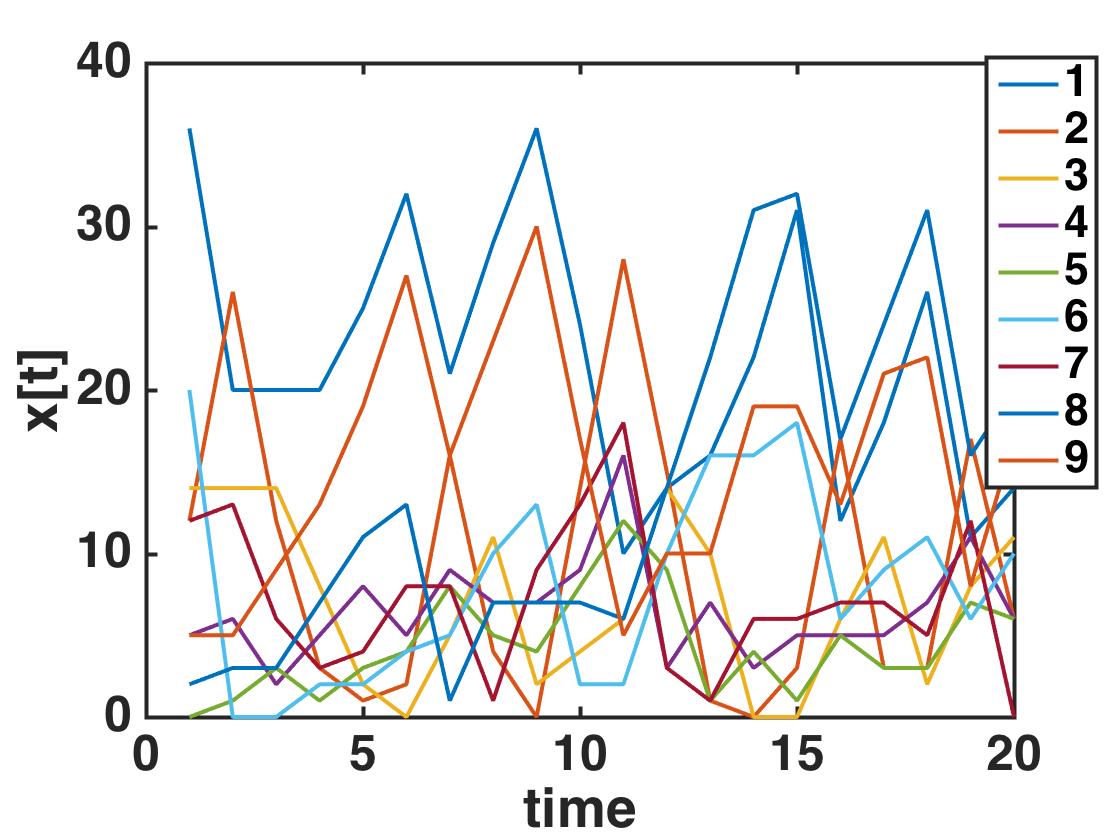} &  \includegraphics[width=0.49\textwidth]{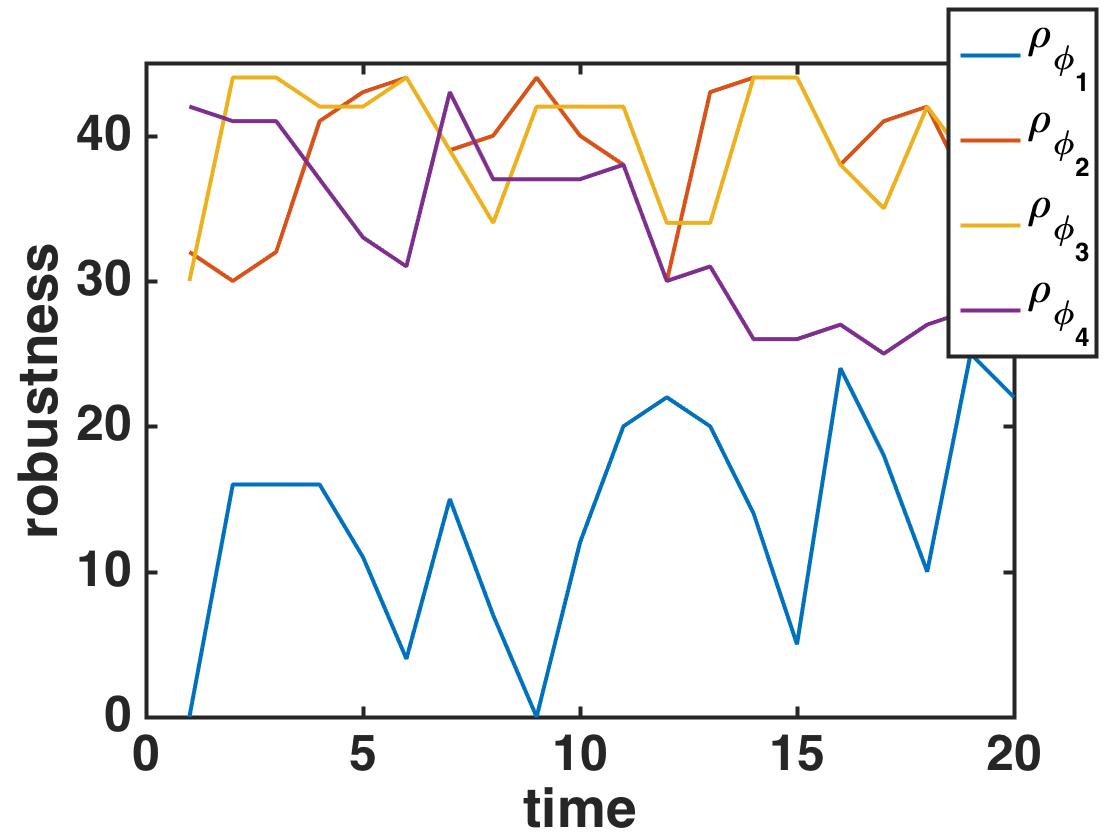}
                        \\
                                          a) &  b) 
        \end{tabular}
\caption{a) The number of vehicles on each link b) robustness $\rho$ measures the distance to the safety set boundaries. Notice that robustness is always above zero since the trajecory is always in the safe set.}

\label{fig:example}
\end{center}

\end{figure}

\section{Future Work}
Our future work involves determining safe sets using learning methods. Using real world traffic data, we will characterize the trajectories that result in undesirable behavior. Therefore, safe sets considered in this paper will not be user-defined but generated from real data. By combining learning techniques, formal methods and control theory, we will improve current measures of intelligent traffic management.

\section*{Acknowledgement}
We thank Majid Zamani and Matthias Rungger from TU M{\"u}nchen for useful discussions on the material of this paper. 

\bibliographystyle{plain}
\bibliography{bib_trafmpc}

\end{document}